\documentclass[11pt]{article}

\usepackage{amssymb,fullpage,tikz,amsmath,amsthm,xspace,algorithm,algpseudocode,colortbl}
\usepackage[utf8]{inputenc}
\usepackage{hyperref}
\hypersetup{colorlinks,allcolors=blue}

\usepackage[url=false]{biblatex}
\newtheorem{theorem}{Theorem}

\newtheorem{lemma}[theorem]{Lemma}
\newtheorem{corollary}[theorem]{Corollary}
\newtheorem{proposition}[theorem]{Proposition}
\newtheorem{definition}[theorem]{Definition}
\newtheorem{example}[theorem]{Example}
\newtheorem{remark}[theorem]{Remark}

\begin{document}

\title{Scenario-Based Robust Optimization of Tree Structures\thanks{Partially supported by the grant ANR-23-CE48-0010 PREDICTIONS from the French National Research Agency (ANR).}}

\author{%
Spyros Angelopoulos\thanks{Sorbonne Université, CNRS, LIP6, Paris, France}
\and
Christoph Dürr\footnotemark[2]
\and
Alex Elenter\footnotemark[2]
\and
Georgii Melidi\footnotemark[2]
}

\date{}
\maketitle

\begin{abstract}
We initiate the study of tree structures in the context of {\em scenario}-based robust optimization. Specifically, we study Binary Search Trees (BSTs) and Huffman coding, two fundamental techniques for efficiently managing and encoding data based on a known set of frequencies of keys. Given $k$ different scenarios, each defined by a distinct frequency distribution over the keys, our objective is to compute a single tree of best-possible performance, relative to any scenario. 

We consider, as performance metrics, the {\em competitive ratio}, which compares multiplicatively the cost of the solution to the tree of least cost among all scenarios, as well as the {\em regret}, which induces a similar, but additive comparison. For BSTs, we show that the problem is NP-hard across both  metrics. We also show how to obtain a tree of competitive ratio  $\lceil \log_2(k+1) \rceil$, and we prove that this ratio is optimal. For Huffman Trees, we show that the problem is, likewise, NP-hard across both metrics; we also give an algorithm of regret $\lceil \log_2 k \rceil$, which we show is near-optimal, by proving a lower bound of $\lfloor \log_2 k \rfloor$. Last, we give a polynomial-time algorithm for computing {\em Pareto-optimal} BSTs with respect to their regret, assuming scenarios defined by uniform distributions over the keys. This setting captures, in particular, the first study of {\em fairness} in the context of data structures. We provide an experimental evaluation of all algorithms. To this end, we also provide  mixed integer linear program formulation for computing optimal trees. 
\end{abstract}

\paragraph*{keywords}	
	Binary Search Tree, Huffman Tree, Pareto optimality, robust optimization, scenarios, fairness. 

\section{Introduction}
\label{sec:introduction}

Suppose that we would like to encode and transmit a text in a given language efficiently, i.e., using the least number of bits on expectation. If the alphabet's frequency is known ahead of time, i.e., if the language is pre-determined, this can be done efficiently using the well-known technique of {\em Huffman coding}~\cite{huffman1952method}. But what if we do not know in advance the intended language, but instead it is only known that it can be either English, Italian, or Finnish? In this case, one would like to design a {\em single} code that performs efficiently across all three such  scenarios, and in particular against a worst-case, adversarially chosen language. 
  
For a second example, suppose that we are given $n$ keys to be stored in a {\em Binary Search Tree} (BST), and that we would like to minimize the expected number of comparisons, when performing a search operation for a key. Once again, if the key frequencies are known ahead of time, finding an {\em optimal} BST is a fundamental problem in Computer Science, going back to the seminal work of Knuth~\cite{knuth1971optimum} who gave a quadratic-time algorithm. Suppose, however, that instead of a single frequency vector, we are provided with $k$ different vectors, each defining one possible scenario. How would we construct a {\em single} BST that performs well across all scenarios, hence also against one chosen adversarially? 

Motivated by the above situations, in this work we introduce the study of {\em robust} structures, in the presence of $k$ possible {\em scenarios}. Each scenario is described by means of a frequency vector over the keys, and we seek a single tree that is robust with respect to {\em any} scenario (and thus with respect to the worst-case, adversarial scenario). This approach falls into what is known as {\em scenario-based robust optimization}; see~\cite{ben2009robust} for a survey of this area. While the scenario-based framework  has been studied in the context of AI-related optimization problems in areas such as  planning~\cite{mccarthy2017staying}, scheduling~\cite{shabtay2023state} and network optimization~\cite{kasperski2015complexity,kasperski2015approximability}, to our knowledge it has not been applied to data structures or data coding.

We consider two main measures for evaluating the quality of our solutions. The first measure applies to BSTs, and is the worst-case ratio, among all scenarios $s$,  between the cost of our solution (tree) under scenario $s$ and the cost of the optimal tree under $s$; using the canonical term from online computation, we refer to this measure as the {\em competitive ratio}. For Huffman trees, we consider even stronger guarantees, by studying the worst-case {\em difference}, among all scenarios $s$ between the cost of our solution under $s$ and the cost of the optimal tree for $s$, namely the {\em regret} of the solution. We refer to Sections~\ref{subsec:bst.background} and~\ref{subsec:ht.background} for the formal definitions. Competitive analysis and regret minimization are both well-studied approaches in optimization under uncertainty, that establish strict, worst-case guarantees under adversarial situations~\cite{borodin2005online},~\cite{blum2007learning}. Competitive analysis is the predominant analysis technique for tree-based data structures, see the seminal work~\cite{sleator1985self} on Splay Trees. Regret minimization, on the other hand, may provide more stringent guarantees (c.f. Corollary~\ref{cor:ht.comp}), and notions related to regret have been applied, for instance, to the evaluation of query efficiency in databases~\cite{xie2020experimental},~\cite{nanongkai2010regret}.

\subsection{Contribution}
\label{subsec:contributions}

We begin with the study of robust BSTs in Section~\ref{sec:bst}. We first show that minimizing either the competitive ratio or the regret is NP-hard, even if there are only two scenarios, i.e., $k=2$ (Theorem~\ref{thm:bst.nphard}). We also give an algorithm that constructs a BST of competitive ratio at most $\lceil \log_2 (k+1) \rceil$ (Theorem~\ref{th:alg.bst.ratio}), and we show that this bound is optimal, in that there exists a collection of $k$ scenarios under which no BST can perform better (Theorem~\ref{prop.bst.lb}).

In Section~\ref{sec:huff}, we study robust Huffman trees (HTs). We first show that the problem of minimizing the competitive ratio or the regret is NP-hard, again even if $k=2$ (Theorem~\ref{thm:ht.nphard}). We also give an algorithm for constructing a Huffman tree that has regret that is provably at most  $\lceil \log_2 k \rceil$. We show that this is essentially optimal, by proving a near-matching lower bound equal to $\lfloor \log_2 k \rfloor$ (Theorem~\ref{thm.ht.lb}).

In Section~\ref{sec:regret.bst}, we study the problem of minimizing BST regret with respect to $k$ scenarios. This problem can be formulated as a $k$-objective optimization problem, hence we seek trees in the {\em Pareto-frontier} of the $k$-objectives. For concreteness, we focus on scenarios induced by uniform distributions over subsets of keys: here, we give a polynomial-time algorithm for finding Pareto-optimal solutions. This formulation provides the first framework for quantifying {\em fairness} in the context of data structures. More precisely, we can think of each distinct scenario as the profile of a different {\em user}, and the BST as the single {\em resource} that is shared by all $k$ users. We thus seek a solution that distributes the cost as equitably as possible among the $k$ competing users.       

In Section~\ref{sec:experiments}, we provide an experimental evaluation of all our algorithms over real data. To be able to evaluate and compare the algorithms from Section~\ref{sec:bst} and Section~\ref{sec:huff}, we also provide mixed integer linear program formulations for all objectives, which allows us to compute the optimal trees for small instances. 

In terms of techniques, despite the seeming similarity of the settings, we show that robust BSTs and HTs are quite different problems; this is due to the fact that in the former the keys are stored in all nodes, whereas in the latter they are stored in leaves. This is reflected both in the different NP-hardness proofs and in the different algorithmic approaches. Namely, for BSTs, NP-hardness is established using a non-trivial reduction from the \textsc{Partition} problem (Theorem~\ref{thm:bst.nphard}), whereas for HTs we use a non-trivial reduction from the \textsc{Equal-Cardinality Partition} problem (Theorem~\ref{thm:ht.nphard}). From the algorithmic standpoint, for BSTs we follow a recursive approach for constructing the tree, whereas for HTs we use an approach that allows us to ``aggregate'' the optimal HT for each scenario into a single HT. Last, in terms of finding the Pareto-frontier for regret minimization in BSTs, we use a dynamic programming approach that allows for an efficient implementation.    

\subsection{Related Work}
\label{subsec:related}

The BST is a fundamental data structure that has been studied extensively since the 1960s. See~\cite{windley1960trees, booth1960efficiency, hibbard1962some} for some classic references and~\cite{NAGARAJ19971} for a survey. Given a frequency vector of  accesses to $n$ different keys, finding a BST of optimal average access cost was originally solved in~\cite{knuth1971optimum}. Likewise, Huffman coding~\cite{huffman1952method} is a fundamental technique for lossless data compression, which combines optimality of performance and simplicity of implementation. Given a frequency vector over an $n$-sized alphabet, a Huffman tree can be implemented in time $O(n\log n)$ using a priority queue. We refer to~\cite{moffat2019huffman} for a survey. 

A recent, and very active related direction in Machine Learning seeks to augment data structures with some {\em prediction} concerning the future access requests. Examples of learning-augmented data structures that have been studied include skip lists~\cite{fu2024learning,zeynali2024robust}, BSTs~\cite{lin2022learning,cao2022learning}, B-trees~\cite{kraska2018case} and rank/select dictionaries~\cite{boffa2022learned}; see also~\cite{Ferragina2020} for a survey on the applications of ML in data structures. These works leverage a learned prediction about access patterns, and seek structures whose performance degrades smoothly as a function of the prediction error. The settings we study in this work can thus be interpreted, equivalently, as having access to $k$ different predictions, and seeking structures that perform efficiently even if the worst-case prediction materializes. 

It is known that if the access frequencies are chosen uniformly at random, then with high probability the optimal BST is very close to a complete binary search tree, and its expected height is logarithmic in the number of keys $n$. In the other extreme, there exist adversarial frequencies for which the optimal BST has height that is as large as $\Omega(n)$. Several works have studied the regime between these extremes, via small, random perturbations to the frequency vector, e.g.,~\cite{manthey2007smoothed}. Bicriteria optimization problems over BSTs were studied in~\cite{mankowski2020dynamic}, with the two objectives being the maximum and the average weighted depth, respectively.

Scenario-based robust optimization has been extensively applied to the study of {\em scheduling} problems under uncertainty. Examples include the minimization of completion time~\cite{mastrolilli2013single}, flow-shop scheduling~\cite{kasperski2012approximating} and just-in-time scheduling~\cite{gilenson2021multi}. We refer to~\cite{shabtay2023state} for a recent survey of many results related to scenario-based scheduling.

We conclude with a discussion of {\em fairness}, which is becoming an increasingly demanded requirement in algorithm design and analysis. Algorithmic fairness, defined as an equitable treatment of users that compete for a common resource,  has been studied in some classic settings, including optimal stopping problems~\cite{arsenis2022individual},~\cite{buchbinder2009secretary}, and resource allocation problems such as knapsack~\cite{lechowicztime, patel2021group}. However, to our knowledge, no previous work has addressed the issue of fairness in the context of data structures, although the problem is intrinsically well-motivated: e.g., we seek a data structure (such as a BST) that guarantees an equitable search time across the several competing users that may have access to it. The concept of Pareto-dominance as a criterion for fairness has attracted attention in recent works in ML, e.g.,~\cite{martinez2021blind},~\cite{martinez2020minimax}. In our work, we do rely on learning oracles, and we seek Pareto-optimal solutions that capture group fairness based on regret.
\section{Robust Binary Search Trees}
\label{sec:bst}

\subsection{Background and Measures}
\label{subsec:bst.background}

In its standard form, a BST stores $n$ keys from a given ordered set; without loss of generality, we may assume that the set of keys is the set $\{1, \ldots, n\}$. The  keys are stored in the nodes, and satisfy the \emph{ordering} property: the key of any node is larger than all keys in its left sub-tree, and smaller than all keys in its right sub-tree. For a given key, the corresponding node is accessed in a recursive manner, starting from the {\em root}, and descending in the tree guided by key comparisons. 

A BST can be conveniently represented, equivalently, by its {\em level vector}, denoted by $L$, in the sense that every key $i$ has level $L_i$ in the tree. Here we use the convention that the root has level $1$. Formally, we have:

\begin{definition} \label{level_definition}
    A vector $L\in \{1,\ldots,n\}^n$ is a \emph{level vector} of a BST if and only if for every $1\leqslant i<j\leqslant n$ with $L_i=L_j$, there is a key $i<r<j$ such that $L_r < L_i$. 
\end{definition}

The definition formulates the fact that  if there are two keys at the same level of the BST, then there must exist a key between them of lower level (i.e., higher in the tree). This definition allows us to express the {\em average cost} of a BST represented by a level vector $L$, relative to a frequency vector $F$, as the inner product
\begin{equation}
\textrm{cost}(L,F)= \sum_{i=1}^n L_i\cdot F_i.
\label{eq:vector-def}
\end{equation}

Given a frequency vector $F$, a level vector $L$ minimizing $\textrm{cost}(L,F)$ can be computed, using dynamic programming, in time $O(n^2)$~\cite{knuth1971optimum}. 

Next, we define formally the \emph{robust BST} problem. Its input consists of $k$ frequency vectors $F^1,\ldots,F^k$, called {\em scenarios}. There are three possible performance metrics one could apply in order to evaluate the performance, which give rise to three possible minimization problems on the BST level vector $L$ that must be found:

\smallskip 

\noindent
$\bullet$
{\em Worst-case cost}: here, the objective is to minimize $\max_s \textrm{cost}(L,F^s)$, i.e.,  we seek the tree of smallest cost under the worst-case scenario.

\smallskip
\noindent
$\bullet$ \ 
{\em Competitive ratio}: here, we aim to minimize
\[\max_s \frac{\textrm{cost}(L,F^s)}{ \min_{L^*} \textrm{cost}(L^*, F^s)},\]
i.e., we aim to minimize the worst-case ratio between the cost of our tree and the optimal tree for each scenario.

\smallskip
\noindent
$\bullet$ \ 
{\em Regret}: here, the objective is to minimize the quantity 
\[
\max_s \{ \textrm{cost}(L,F^s) - \min_{L^*} \textrm{cost}(L^*, F^s) \},\]
i.e., we want to achieve the smallest {\em difference} between the cost of our tree and the optimal tree for each scenario.

Note that if for each scenario the respective optimal trees have the same cost, then the three metrics are equivalent. However, in general, they may be incomparable, as shown in the following example. 
 
\begin{example}
Let $F^1=(0,1/4,3/4)$ and $F^2=(4/9,2/9,1/3)$ denote two scenarios for three keys $\{a,b,c\}$. The various metrics for each  possible BST are as depicted.
    \begin{quote}
    \begin{tabular}{lccccc}
    &
    \begin{tikzpicture}[scale=0.8]
        \node (a) at (-0.7,0) {a};
        \node (b) at (0,-0.7) {b};
        \node (c) at (0.7,-1.4) {c};
        \draw (a) -- (b) -- (c);
    \end{tikzpicture}
    &
    \begin{tikzpicture}[scale=0.8]
        \node (a) at (-0.7,0) {a};
        \node (b) at (0,-1.4) {b};
        \node (c) at (0.7,-0.7) {c};
        \draw (a) -- (c) -- (b);
    \end{tikzpicture}
    &
    \begin{tikzpicture}[scale=0.8]
        \node (a) at (-0.7,-0.7) {a};
        \node (b) at (0,-1.4) {b};
        \node (c) at (0.7,0) {c};
        \draw (c) -- (a) -- (b);
    \end{tikzpicture}
    &
    \begin{tikzpicture}[scale=0.8]
        \node (a) at (-0.7,-0.7) {a};
        \node (b) at (0,0) {b};
        \node (c) at (0.7,-0.7) {c};
        \node (d) at (0, -1.4) {};
        \draw (a) -- (b) -- (c);
    \end{tikzpicture}
&	
\begin{tikzpicture}[scale=0.8]
    \node (a) at (-0.7,-1.4) {a};
    \node (b) at (0,-0.7) {b};
    \node (c) at (0.7,0) {c};
    \draw (a) -- (b) -- (c);
\end{tikzpicture}
\\
cost for $F^1$& 11/4  & 9/4   & 7/4   & 3/2   & \cellcolor{lightgray}5/4 \\
cost for $F^2$& 17/9  & \cellcolor{lightgray}16/9  & 16/9  & 17/9  & 19/9 \\
worst cost & 11/4  & 9/4   &\cellcolor{yellow} 16/9  & 17/9  & 19/9 \\
competitive ratio & 11/5  & 9/5   & 7/5   & 6/5   & \cellcolor{yellow}19/16 \\
regret & 3/2   & 1     & 1/2   &\cellcolor{yellow} 1/4   & 1/3\\ 
\end{tabular}
\end{quote}

The first two rows of the table show the cost of each of the five possible BSTs on three keys, for the two scenarios $F^1$ and $F^2$; here the optimal costs for each scenario are highlighted in gray. The remaining three columns show the performance of each tree with respect to the three metrics, with the best performance  highlighted in yellow.

\end{example}
 
 Our hardness results, for both BSTs and HTs, apply to all three metrics. However, for the purposes of analysis, we focus on the competitive ratio and the regret. As discussed in Section~\ref{sec:introduction}, the competitive ratio is the canonical performance notion in the analysis of BSTs, and allows us to capture worst-case performance under uncertainty that can be efficiently approximated, both from the point of view of upper bounds (positive results) and lower bounds (impossibility results). In other words, the competitive ratio reflects the {\em price} of not knowing the actual scenario in  advice, similar to the competitive analysis of online algorithms~\cite{borodin2005online}. On the other hand, regret-minimization can provide even stronger guarantees for HTs, but can also help model issues related to fairness in BSTs, as we discuss in detail in Section~\ref{sec:regret.bst}.

\subsection{Results}
\label{subsec:bst.results}

We first show that finding an optimal robust BST is NP-hard, even if $k=2$. We will prove the result for the cost-minimization version; however the proof carries over to the other metrics. 

To prove NP-hardness, we first need to formulate the decision variant of the problem:
Given two frequency vectors, $F^1$ and $F^2$, and a threshold $V$, the objective is to decide whether there exists a BST $L$ of cost at most $V$, in either scenario, i.e., $\textrm{cost}(L,F^1)\leq V$ and 
$\textrm{cost}(L,F^2)\leq V$.

\begin{theorem}
    The robust BST problem is NP-hard, even if $k=2$. This holds for all three metrics, i.e., for minimizing the cost, or the competitive ratio, or the regret.
\label{thm:bst.nphard}   
\end{theorem}

\begin{proof}
The proof is based on a reduction from the \textsc{Partition} problem~\cite{garey1979computers}.
An instance of this problem consists of a list $a_1,...,a_m$ of non-negative integers, and the goal is to decide whether there exists a binary vector $b\in\{0,1\}^m$ with $\sum_i b_i a_i=\sum_i (1 - b_i)a_i$. Without loss of generality we can assume that $m$ is of the form $m = 2^\ell$ for some integer $\ell$, as we can always pad the instance with zeros. 

From the given instance of the partition problem, we define an instance of the robust BST problem which consists of $n=3m-1$ keys and two frequency vectors 
\begin{align*} 
    F^1 &= ( a_1, 0, w, a_2, 0, w, a_3, 0, \ldots, w, a_m,0) \text{ and}\\
	F^2 &= ( 0, a_1, w, 0, a_2, w, 0, a_3, \ldots, w, 0,a_m),
\end{align*}
where $w$ is any constant strictly larger than $(\ell+3/2)\sum_i a_i$. Moreover, the instance specifies the threshold $V=W+ (\sum a_i)/2$, where $W$ is defined as 
\[
W := w\sum_{j=1}^{\ell} j 2^{j-1} + (\ell+1)\sum_{i=1}^n a_i.
\]

We claim that the level vector $L$ which optimizes $\max\{\text{cost}(L,F^1), \text{cost}(L, F^2)\}$ has the following structure. The first $\ell$ levels form a complete binary tree over all keys $3i$ for $i=1,2,\ldots,m$. These are exactly the keys with frequency $w$ in both $F^1$ and $F^2$. (Intuitively, since $w$ is large, these keys should be placed at the smallest levels.) Furthermore, for each $i=1,2,\ldots,m$, the keys $3i-2$ and $3i-1$ are placed at levels $\ell+1$ and $\ell+2$, or at levels $\ell+2$ and $\ell+1$, respectively. Such a tree has a cost at least $W$ in each scenario and at most $W+\sum a_i$. The claim that $L$ has this structure follows from the fact that if some key with frequency $w$ were to be placed below level $\ell$, then the tree would incur a cost of at least $w(1 + \sum_{j=1}^{\ell} j 2^{j-1})$, which is strictly greater than $V$. See Figure~\ref{fig:np-hardness-bst} for an illustration. 

\begin{figure}[htb]
    \centering
    \includegraphics[scale=0.2]{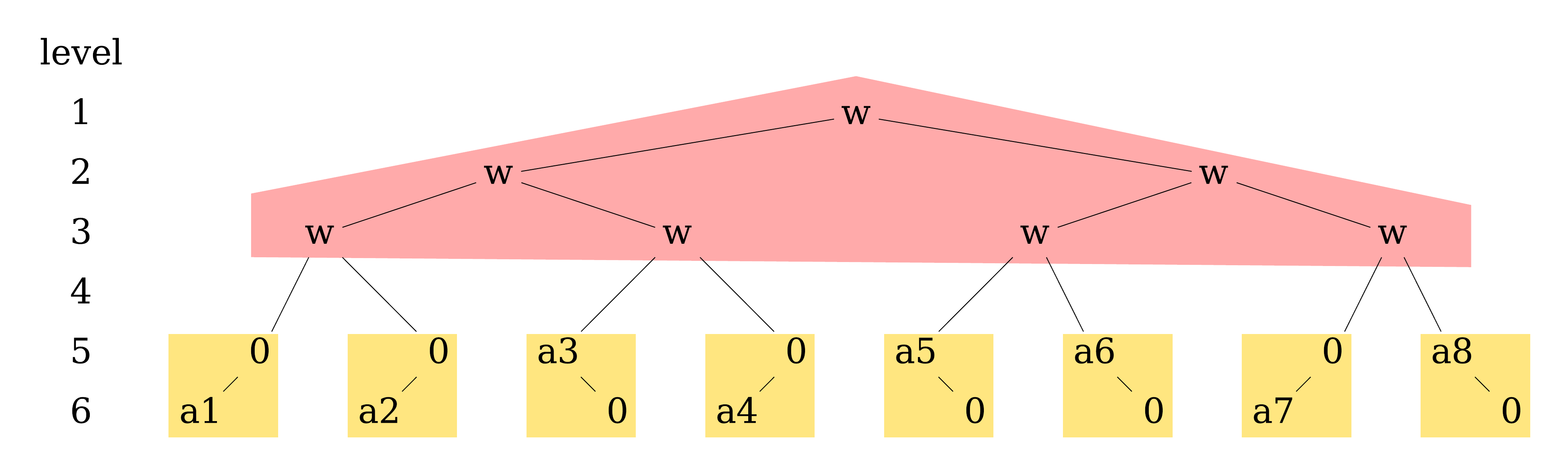}
    \caption{The BST corresponding to the binary vector $b=\{11010010\}$ in the NP-hardness proof construction. Nodes are labeled with the frequency of their keys in $F^1$.}
    \label{fig:np-hardness-bst}
\end{figure}	

As a result, all solutions to the robust BST problem can be described by a binary vector $b\in\{0,1\}^{m}$, such that for every $i=1,2,\ldots,m$ the key $3i-2$ has level $\ell+1+b_i$ while key $3i-1$ has level $\ell+2-b_i$. We denote by $L^b$ the level vectors of these trees in order to express the costs as $\textrm{cost}(L^b,F^1) = W + \sum_i b_i a_i$ and $\textrm{cost}(L^b,F^2) = W + \sum_i (1-b_i)a_i$.

Observe that if both costs of $L^b$ are at most $V$, then they are equal to this value and $b$ is a solution to the partition problem.  This implies that deciding if there exists a BST of cost at most $V$ is NP-hard. 

The proof extends to the other two metrics, namely the competitive ratio, and the regret. Observe that level vector $L^b$ for $b=(0,\ldots,0)$  minimizes the cost for the first scenario, while for $b=(1,\ldots,1)$ it minimizes the cost for the second scenario, and that both costs are equal to $W$. Hence, it is also NP-hard to decide if there is a BST of competitive ratio at most $(W+(\sum a_i)/2) / W$. Similarly, it is NP-hard to decide if there is a BST of regret at most $(\sum a_i)/2$.
\end{proof}	

Next, we present an algorithm that achieves the optimal competitive ratio. Algorithm~\ref{alg:rbst} first computes optimal trees with level vector $L^s$ for each scenario $s=1, \ldots, k$, then calls a recursive procedure $A$ with initial values $i=1$, $j=n$ and $\ell=1$. Procedure $A$ solves a subproblem defined by $i,j,\ell$. Namely, it constructs a BST on the keys $[i,j]$ and stores its root at level $\ell$. To this end, it first identifies the key(s) in $[i,j]$ at the minimum level in the optimal trees across all scenarios $L^s$ (lines 5 and 6) and stores them in the ordered set $S$. If an odd number of keys share this minimum level, the algorithm selects the key to be stored at level $\ell$ as the median of the set $S$. Otherwise, it chooses the key closest to the median value (line 7).

\begin{algorithm}[H]
\begin{algorithmic}[1]
\Statex \textbf{Input:} $k$ scenarios described by frequency vectors $F^1, \ldots, F^k$. 
\Statex \textbf{Output:}  A robust BST represented as a level vector $L$.
\State For each scenario $s \in \{1, \ldots, k\}$, compute optimal tree $L^s$
\State $A(i=1,j=n,\ell=1)$
\end{algorithmic}
\caption{Algorithm {\sc R-bst}}
\label{alg:rbst} 
\end{algorithm}

\begin{algorithm}[H]
\begin{algorithmic}[1]
\Procedure{A}{$i,j,\ell$}
    \If{$i>j$}
        \State \textbf{return} \Comment{Empty interval $[i,j]$ ends recursion}
    \Else
        \State Let $v$ be the minimum level $L^s_r$ over $i\leqslant r\leqslant j$ and $1\leqslant s\leqslant k$.
        \State Let $S$ be the ordered set of keys  $i\leqslant r\leqslant j$ such that there exists a scenario $s$ with $L^s_r=v$.
        \State Let $m\in S$ be s.t.\ both $|S\cap [i,m-1]|$ and $|S\cap[m+1,j]|$ are at most $\lceil (|S|-1)/2\rceil$. 
        \State Set $L_m=\ell$
        \State $\Call{A}{i, m-1, \ell+1}$ \Comment{Left recursion}
        \State $\Call{A}{m+1,j,\ell+1}$  \Comment{Right recursion}	
    \EndIf
\EndProcedure
\end{algorithmic}
\caption{Recursive procedure $A$}
\label{algorithm_A}
\normalsize
\end{algorithm}

We show that {\sc R-bst} has competitive ratio  logarithmic in~$k$:

\begin{theorem} \label{th:alg.bst.ratio}
    For every scenario $s$, {\sc R-bst} constructs a BST of cost at most $\lceil \log_2(k+1) \rceil$ times the cost of the optimal BST $L^s$. 
\end{theorem}
\begin{proof}
Let $R=\lceil \log_2(k+1) \rceil$ be the claimed competitive ratio. We will show that in line 5 of Algorithm \ref{algorithm_A}, we have $\ell \leqslant R v$. This implies that for every key $a$, and every scenario $s$, we have $L_a\leqslant R \cdot L^s_a$, which suffices to prove the theorem.

First, we observe by the definition of level vectors that for every scenario $s$, there is at most one key $i\leqslant r\leqslant j$ with $L^s_r=v$. This implies that $|S|\leqslant k$. Furthermore, with each recursive call, the level $\ell$ increases by $1$, while either $v$ increases as well, or $|S|$ decreases to at most $\lceil(|S|-1)/2\rceil$.
It remains to show that after at most $R$ recursive calls, the value $v$ has strictly increased. 

To this end, given some $x$, define $h(x)$ as the largest value such that $\lceil(h(x)-1)/2\rceil=x$. This function is $h:x\mapsto 2x+1$. The sequence $0,h(0),h(h(0)), \ldots$ is thus a sequence of numbers of the form $2^i-1$, with $i \in \mathbb{N}$. Hence the maximum number of recursions needed until $v$ strictly increases is the smallest $x$ such that $k\leqslant 2^x - 1$, which is precisely $R$.
\end{proof}

We can implement {\sc R-bst} so that it runs in time $O(kn^2)$, i.e., the run time is dominated by the time required to compute the optimal BSTs for each scenario. 

\begin{theorem}
{\sc R-bst} can be implemented in time 
$O(kn^2)$.
\label{prop:runtime.rbst}
\end{theorem}
\begin{proof}
Computing the level vectors of the optimal BSTs for each scenario requires time $O(k n^2)$~\cite{knuth1971optimum}. From this, we can compute a vector $M$ with $M[j]=\min_s L^s_j$, in time $O(kn)$. In addition, we use another table $P$, containing lists entries to keep track of the positions of each value in $M$. Formally, $P[v]$ is the list of all positions $j$ with $M[j]=v$. Table $P$ is constructed by a simple left-to-right processing of $M$, and as a result all the lists are sorted. During this process, we also create a table $Q$ with ranks, such that $Q[j]$ is the number of entries in $M$ between index $1$ and $j$ which has the value $M[j]$. $P$ and $Q$ are complementary in the sense that the $i$-th entry $j$ in the list $P[v]$, is such that $M[j]=v$ and $Q[j]=i$. Finally, we store the arrays $M$ and $Q$ in a single {\em segment tree}, which allows to return for a given index interval $[i,j]$, the minimum value $v$ in $M$ between indices $i$ and $j$, together with the ranks of the first occurrence of $v$ and the last one in this interval. Let $s,e$ be these ranks. 

Line 7 of Algorithm~\ref{algorithm_A} can be implemented by querying the segment tree at the interval $[i,j]$, and receiving $v,s,e$ as a result. Then the algorithm chooses $m$ as the $\lfloor (s+e)/2 \rfloor$-th element in the list $P[v]$. Each query to the segment tree takes time $O(\log n)$. Since the resulting tree has $n$ nodes, the construction of the tree takes time $O(n\log n)$, which is also the pre-processing time of the segment tree. The total running time of the algorithm is therefore dominated by the time to build the level vectors $L^s$, i.e., at most $O(k n^2)$.
\end{proof}

We conclude this section by showing that 
{\sc R-bst} achieves the optimal competitive ratio.
\begin{theorem}
    There exists a collection of $k$ scenarios, described by vectors $F^1, \ldots ,F^k$, such that for every BST with level vector $L$, there exists a scenario $s$ for which  $\textrm{cost}(L,F^s)\geqslant \lceil \log_2(k+1) \rceil \textrm{cost}(L^s, F^s)$.
\label{prop.bst.lb}
\end{theorem}
\begin{proof}
Consider the following $k$ scenarios $F^1, \ldots, F^k$ on $n=k$ keys, for some $k$ of the form $2^\ell-1$, defined as
\begin{align*}
	F^1  &= (1,0,  \ldots ,0), \\ 
    F^2  &= (0,1, \ldots, 0), \\ 
    &\ldots \\
    F^k  &= (0,0, \ldots, 1).
\end{align*}

Note that for each scenario, the corresponding optimal tree has cost $1$.
Since the complete tree with $\ell$ levels has $2^{\ell}-1$ nodes, any tree with $k$ nodes has at least one key at level $\lfloor \log_2(k+1)\rfloor$, and therefore cost at least $\ell$ in one of the scenarios. Hence, for every tree level vector $L$, there exists $s$ such that $\textrm{cost}(L,F^s)\geqslant \lceil \log_2(k+1) \rceil \textrm{cost}(L^s, F^s)$.
\end{proof}

\section{Robust Huffman trees}
\label{sec:huff}

\subsection{Background and Measures}
\label{subsec:ht.background}

In the standard version of the Huffman tree problem, we are given $n$ keys (e.g., letters of an alphabet), along with a frequency vector (e.g., the frequency of each letter in said alphabet). The objective is to build a binary tree in which every key corresponds to a {\em leaf}, so as to minimize the inner product $\sum_i F_i L_i$, where $L_i$ is the level of key $i$ in the tree. By labeling the edges leaving each inner node with $0,1$, arbitrarily, we can associate a \emph{codeword} to each key, namely the concatenation of the edge labels on the path from the root to the leaf corresponding to the key in question. In this setting, we assume that the levels of the tree start from 0. Hence, conveniently, the length of a codeword equals the level of the corresponding leaf. The resulting set of codewords is \emph{prefix free}, in the sense that no codeword is a prefix of another codeword. The optimal tree can be computed in time $O(n \log n)$~\cite{huffman1952method}.

We study Huffman coding in a robust setting in which we are given $k$ $n$-dimensional frequency vectors $F^1,\ldots,F^k$, each describing a different scenario. We say that a level vector $L$ is \emph{valid} if there is a prefix-free code such that key $i$ has a codeword of length $L_i$. In other words, for any valid level  vector $L$, there is a Huffman tree with $|\{i: L_i = a\}|$ leaves on level $a$. The cost of $L$ in a given scenario $s$ is defined as the inner product $F^s \cdot L = \sum_i F^s_i L_i$.

We can analyze robust HTs using the same measures as for robust BSTs, i.e., using cost-minimization, competitive ratio, or regret-based analysis; recall the discussion in Section~\ref{subsec:bst.background}. For Huffman trees, in particular, we will rely on regret-based analysis, which establishes more refined performance guarantees than competitive ratio; c.f. Corollary~\ref{cor:ht.comp} which shows that our results for regret essentially carry over to the competitive ratio.

\subsection{Results}
\label{subsec:ht.results}

We begin by showing that finding an optimal robust HT is NP-hard, even in the case of only two scenarios. The proof differs substantially from the NP-hardness proof for robust BSTs (Theorem~\ref{thm:bst.nphard}). This is due to the differences in the two settings: in a BST, keys are stored in each node, whereas in a HT, keys are stored only at leaf nodes. As a result, the reduction is technically more involved, and is from a problem that induces more structure, namely the 	\textsc{Equal-Cardinality Partition} problem \cite{garey1979computers}.

We first formulate the decision variant of the robust HT problem:  Given two frequency vectors, $F^1$ and $F^2$, and a threshold $V$, we must decide whether there exists a HT such that its cost does not exceed $V$ in either scenario.

The following theorem relies on a proof that differs substantially from the corresponding NP-hardness proof for robust BSTs (Theorem~\ref{thm:bst.nphard}). As discussed previously, these two problems differ significantly, primarily because the keys are located on the leaves in a Huffman tree. This structural difference forces us to establish entirely new properties to achieve the reduction and complete the NP-hardness proof. For the former, we need to use a different variant of the partition problem with equal cardinalities to satisfy all the conditions of the Huffman tree and achieve the desired construction.

\begin{theorem} \label{thm:ht.nphard}
	The robust HT problem is NP-hard, even if $k=2$.
    This holds for all three metrics, i.e.,
    for minimizing the cost, the competitive ratio, or the regret.
\end{theorem}
\begin{proof}
We first formulate the decision variant of the robust HT problem:  Given two frequency vectors, $F^1$ and $F^2$, and a threshold $V$, we must decide whether there exists a HT such that its cost does not exceed $V$ in either scenario.

We will prove the result for the cost-minimization metric. At the end of the proof, we argue that it applies straightforwardly to the competitive ratio and the regret minimization as well.

	The proof is based on a reduction from the
	\textsc{Equal-Cardinality Partition} (ECP)
	problem \cite{garey1979computers}. An instance to this problem consists of a list $a_1,...,a_{2n}$ of non-negative integers with $\sum_i a_i = 2D$, and the goal is to decide whether there exists a binary vector $b\in\{0,1\}^{2n}$ such that $\sum_i b_i a_i=\sum_i (1 - b_i)a_i = D$ and in addition $\sum_i b_i = n$. Without loss of generality, we can assume that $n$ is of the form $n = 2^m$ for some integer $m$, as we can always pad the instance with zeros. 

	From the given instance of ECP, we define an instance of the robust HT problem, which consists of $3n+1$ keys, a threshold $V$ (to be determined later in the proof) and two scenarios. The keys are $v_1,\ldots,v_{2n}$ and $u_1,\ldots,u_n,u'$, having the following frequencies in the two scenarios:
    
    \begin{center}
    \begin{tabular}{l|l|l|l|}
        frequency  & $v_i$        & $u_j$ & $u'$ \\ \hline 
        scenario \textbf{1} & $M+ a_i$     &  $M $  & $C_1$ \\
        scenario \textbf{2} & $M- a_i$ &  $M $  & $C_2$ \\
    \end{tabular}
    \end{center}
    
    Here, $C_1, C_2$ and $M$ are constants that will be specified later. Their role is to ensure that key $u'$ is at level $1$ for the robust HT, while also guaranteeing that the optimal trees for scenario \textbf{1} and for scenario \textbf{2} have the same cost. This can be achieved by a sequence of calculations. First, we can assume that both optimal trees have the same structure, meaning that key $u'$ is at level 1 and there are exactly $n$ keys at level $m+2$ and $2n$ keys at level $m+3$. We can achieve this by choosing a large enough constant $M$, which will also be specified later. Next, we focus on the constants $C_1$ and $C_2$, assuming that the list of integers $a_1, a_2, ..., a_{2n}$ is ordered from largest to smallest. We define the sums of the $n$ smallest and $n$ largest values as $X$ and $Y$, respectively, and use them in the following expressions:
    \begin{align*}
        \textrm{cost}(L,F^1) &= C_1 + M(2n(m+3)+n(m+2)) + (m+2)Y + (m+3)X, \textrm{ and} \\
        \textrm{cost}(L,F^2) &= C_2 + M(2n(m+3)+n(m+2)) - (m+2)X - (m+3)Y.
    \end{align*}
    Since we know that $X+Y=2D$, we obtain the desired property $\textrm{cost}(L,F^1) = \textrm{cost}(L,F^2)$ by setting $C_2 = C_1 + 2D(2m+5)$. To simplify our notation, let us define $z:=2n(m+3)+n(m+2)$, as this expression will be used frequently later in the proof.
    
    To establish the reduction, we will show that there is a HT on these $3n+1$ keys with cost at most 
    \[
        V := C_1 + zM + D(2m+5)
    \] 
    in each scenario if and only if there is a solution $b$ to the partition problem.

    For the one direction of the proof, let $b$ be the solution to the partition problem. We construct a HT, in which for every $i=1,\ldots,2n$, the key $v_i$ is located at level $m+2$ if $b_i=0$ and at level $m+3$ if $b_i=1$. All $n$ keys $u_j$ are at level $m+3$ and key $u'$ is at level 1. See Figure~\ref{fig:huffman-np-hardness} for an illustration.
    
    In scenario \textbf{1} the cost of the tree is 
    \begin{align*}
        & C_1 + zM + (m+2)\sum_i a_i(1- b_i)  + (m+3)\sum_i a_i b_i =
        \\
        & C_1 + zM + (m+2)\sum_i a_i  + \sum_i a_i b_i = \\
        & C_1 + zM + (m+2)2D  + D.
    \end{align*}
    
    In scenario \textbf{2} the cost of the tree is 
    \begin{align*}
        & C_2 + zM - (m+2)\sum_i a_i(1- b_i)  - (m+3)\sum_i a_i b_i =
        \\
        & C_1 + zM - (m+2)2D  - D.
    \end{align*}
    By the choice of $C_2$, both costs are equal, and are also equal to $V$.
    
	For the other direction of the proof, let $L$ be the level vector of a HT with cost at most $V$ in each scenario. First, we claim that key $u'$ has to be at level 1 by the assumption that $C_1$ and $C_2$ are huge constants. Let $u''$ be the \emph{sibling} of $u'$, i.e. the other node at level $1$. Second, we claim that the sub-tree rooted at $u''$ is \emph{complete}, in the sense that there are exactly $2n$ keys at level $m+3$ and $n$ keys at level $m+2$. This structure, again as depicted in Figure~\ref{fig:huffman-np-hardness}, is imposed by the contribution of $M$ in the frequencies as we argue next. Consider a Huffman tree $T$ on $3n+1$ keys, all with uniform frequencies $M$ except for key $u'$ which has frequency $C$. If $T$ has the above mentioned structure, then it has cost $C+zM$. If $u'$ is not at level $1$, its cost is at least $2C$. Moreover, if $u'$ is at level $1$ and the sub-tree rooted at the sibling of $u'$ is not complete then its cost is at least $M$ higher. Choosing $M$ large enough such that $M > D(2m+5)$, we have that a tree with cost at most $V$ is complete as well.
	
	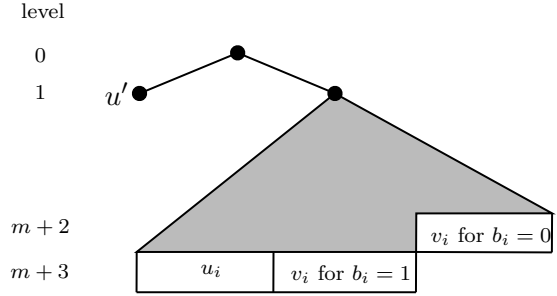
\begin{figure}[htb!]
    \centering
    \tikzset{every picture/.style={line width=0.75pt}} 
    
    \begin{tikzpicture}[x=0.75pt,y=0.75pt,yscale=-1,xscale=1]
    
    \draw  [fill={rgb, 255:red, 155; green, 155; blue, 155 }  ,fill opacity=0.68 ] (308.58,160.5) -- (240.09,160.54) -- (240,180) -- (168.08,180) -- (99.08,180) -- (199.08,100) -- cycle ;
    \draw    (100.5,100) -- (150,79.5) ;
    \draw [shift={(100.5,100)}, rotate = 337.5] [color={rgb, 255:red, 0; green, 0; blue, 0 }  ][fill={rgb, 255:red, 0; green, 0; blue, 0 }  ][line width=0.75]      (0, 0) circle [x radius= 3.35, y radius= 3.35] node[left] {$u'$}  ;
    \draw    (150,79.5) -- (199.08,100) ;
    \draw [shift={(199.08,100)}, rotate = 22.67] [color={rgb, 255:red, 0; green, 0; blue, 0 }  ][fill={rgb, 255:red, 0; green, 0; blue, 0 }  ][line width=0.75]      (0, 0) circle [x radius= 3.35, y radius= 3.35]   ;
    \draw [shift={(150,79.5)}, rotate = 22.67] [color={rgb, 255:red, 0; green, 0; blue, 0 }  ][fill={rgb, 255:red, 0; green, 0; blue, 0 }  ][line width=0.75]      (0, 0) circle [x radius= 3.35, y radius= 3.35]   ;
    \draw     ;
    \draw    (99.08,180) -- (99,200.33) ;
    \draw    (168.08,180) -- (168,200.33) ;
    \draw    (240,180) -- (239.92,200.33) ;
    \draw    (240,180) -- (308.47,180.22) ;
    \draw    (308.58,160.5) -- (308.47,180.22) ;
    \draw    (99,200.33) -- (239.92,200.33) ;
    
    \draw (39.33,53.23) node [anchor=north west][inner sep=0.75pt]  [font=\scriptsize]  {$\text{level}$};
    \draw (46.17,75.73) node [anchor=north west][inner sep=0.75pt]  [font=\scriptsize]  {$0$};
    \draw (46.17,94.4) node [anchor=north west][inner sep=0.75pt]  [font=\scriptsize]  {$1$};
    \draw (34.33,162) node [anchor=north west][inner sep=0.75pt]  [font=\scriptsize]  {$m+2$};
    \draw (34,185) node [anchor=north west][inner sep=0.75pt]  [font=\scriptsize]  {$m+3$};
    \draw (130,184.9) node [anchor=north west][inner sep=0.75pt]  [font=\scriptsize]  {$u_i$};
    \draw (175.33,185.23) node [anchor=north west][inner sep=0.75pt]  [font=\scriptsize]  {$v_{i} \ \text{for} \ b_{i} =1$};
    \draw (246,165.9) node [anchor=north west][inner sep=0.75pt]  [font=\scriptsize]  {$v_{i} \ \text{for} \ b_{i} =0$};
    \end{tikzpicture}
    \caption{Schematic view of the NP-hardness proof construction for Theorem~\ref{thm:ht.nphard}.}
    \label{fig:huffman-np-hardness}
    \end{figure}

    In addition, we can claim that there have to be exactly $n$ leaves at level $m+2$ and $2n$ leaves at level $m+3$, since we know the structure of the tree from the optimality of the complete tree. By an exchange argument, we ensure that all $u_j$ are at level $m+3$, as their frequencies are strictly less than $v_i$ in both scenarios. Let $b\in\{0,1\}^{2n}$ be a vector of indicator bits, with $b_i=1$ if $v_i$ is at level $m+3$. The total cost of such a tree in scenario \textbf{1} is 
    \[
     C_1 +zM + (m+2)2D + \sum a_i b_i,
     \]
     while for scenario \textbf{2} its cost is
     \[
     C_2 + zM - (m+2)2D - \sum a_i b_i.
     \]
     If both costs are at most $V$, then by the choice of $C_2$ they actually equal $V$. This implies that $D=\sum a_i b_i$, and $b$ is a solution to the \textsc{Equal-Cardinality Partition} problem instance.

The above proof is stated in terms of the cost-minimization measure, however it applies straightforwardly to both the competitive ratio and regret minimization, since the optimal cost in both defined scenarios is the same, namely equal to $V$.       
\end{proof}

We propose and analyze an algorithm called {\sc R-ht} for minimizing $k$-scenario regret. The idea is to aggregate the optimal trees for each scenario into a single HT. The algorithm initially computes optimal HTs for each scenario $s$, denoted as $T^s$ (line 2). For each key $i = 1, \ldots, n$, it identifies the scenario $s$ with the shortest code, denoted as $c_i^s \in \{0,1\}^*$ (line 6). Next, it prepends exactly $\lceil \log_2 k \rceil$ bits to each code $c_i^s$, which represent the scenario $s$ (lines 7 and 8). The algorithm generates the final HT by associating each key with a level equal to  $\lceil \log_2 k \rceil + c_i^s$ (line 10). In line 11, ``compactification'' refers to a process which we describe informally, for simplicity. That is,  while there is an inner node $u$ of the HT with outdegree $1$  (meaning it has a single descendant $v$) we contract the edge $u,v$.

\begin{algorithm}
    \begin{algorithmic}[1]
    \Statex \textbf{Input:} $k$ scenarios described by frequency vectors $F^1, \ldots, F^k$. 
    \Statex \textbf{Output:}  A robust Huffman tree $T$.
    \ForAll{scenarios $s$}
        \State Let $T^s$ be the Huffman tree with minimum cost for frequency vector $F^s$
    \EndFor 
    \State Let $\cal C$ be an empty set
    \ForAll{keys $i$}
        \State Let $s$ be a scenario for which key $i$ has the shortest code $c_i^s \in\{0,1\}^*$
        \State Let $x$ be the encoding of the integer $s-1$ with exactly $\lceil \log_2 k\rceil$ bits
        \State Add $x c_i^s$ to $\cal C$
    \EndFor
    \State Build the Huffman tree $T$ for the prefix free codewords $\cal C$
    \State Compactify $T$
    \State \Return{$T$}
    \end{algorithmic}
    \caption{\textsc{R-ht}.}
    \label{alg:rht}  
\end{algorithm}

\begin{theorem}
\label{thm:huffman-approx}
    Algorithm {\sc R-ht} outputs a tree $T$ with a valid level vector $L$ of regret
    at most $\lceil \log_2 k \rceil$. That is, $F^s \cdot L \leqslant \min_{L^*} F^s \cdot L^* + \lceil \log_2 k \rceil$, for every scenario $s$. 
\end{theorem}
\begin{proof}
    First, we observe that the set of all $s$ binary encodings of the integers $0,\ldots, k-1$ form a prefix-free code. Moreover, we observe that concatenating these strings with subsets of the prefix codes for each scenario yields a code which is again prefix-free.
    
   By construction, for every key $i$, the length of its resulting code does not exceed by more than $\lceil \log_2 k \rceil$ the length of the optimal code among all scenarios. This clearly also holds after compactification, hence the theorem follows.
\end{proof}

\begin{proposition}
Algorithm {\sc R-ht} can be implemented in time $O(k n \log n)$, i.e., its runtime is dominated by the time  required to find optimal HTs for each scenario.
\label{prop:ht.implementation}
\end{proposition}
\begin{proof}
For an efficient implementation of {\sc R-ht}, one can work with vectors of code lengths. If $L^s_i$ is the length of the codeword for key $i$ in the optimal HT for scenario $s$, then before compactification, we construct a code length vector $L$ with $L_i=\lceil \log_2 k\rceil + \min_s L^s_i$ for every key $i$. Then, we build the actual HT level by level, assigning for each level $\ell$ the keys $i$ with $L_i=\ell$ to nodes of the growing tree. The construction is done in time linear in the size of the resulting tree.
\end{proof}

The following result establishes a lower bound for our problem, and shows that {\sc R-ht} is essentially best-possible. 

\begin{theorem}
There exists a set of $k$-scenarios for which no robust HT has regret smaller than $\lfloor \log_2 k \rfloor$.
\label{thm.ht.lb}
\end{theorem}
\begin{proof}
Consider $k$ scenarios defined by frequency vectors $F^1,\ldots,F^k$ on $n=k+1$ keys, with $F^s_i=1$ if $s=i\in\{1,\ldots,n\}$ and $F^s_i=0$ otherwise. Note that $k+1$ is a dummy key which has zero probabilities in every scenario, but will be useful in the construction. Clearly, we have that $\min_{L^*} F^s \cdot L^*=1$ for every scenario $s$. By a simple counting argument, every valid level vector on $n$ keys must assign level at least $\lceil \log_2 n \rceil$ to at least one key. Moreover, since inner nodes have out-degree 2, there must exist at least two keys at level at least $\lceil \log_2 n \rceil$. Let $s$ be one of those keys, with the property $s\in\{1,\ldots,k\}$, and with a slight abuse of notation, let $s$ also denote the corresponding scenario (recall that each scenario has support 1 in the given construction).  It follows that the cost of the tree for scenario $s$ is at least $\lceil \log_2 n \rceil$, and thus the difference from the optimal cost of 1 for scenario $s$ is at least  $\lceil \log_2 n \rceil - 1$. If $k$ is a power of 2 we have 
\[
	\lceil \log_2 n \rceil - 1 = \log_2 k + 1 - 1 = \lfloor \log_2 k \rfloor,
\]
  whereas if $k$ is not a power of two we also have
\[
	\lceil \log_2 n \rceil - 1 = \lceil \log_2 k \rceil - 1 = \lfloor \log_2 k \rfloor,
\]
which completes the proof.
\end{proof}

We conclude this section by noting that the guarantees we established on the regret of HTs essentially carry over to the competitive analysis of the trees. 

\begin{corollary}
There exists an algorithm for $k$-scenario robust HTs of competitive ratio at most $\lceil \log k \rceil+1$. Furthermore, no algorithm for this problem can have competitive ratio better than $\lfloor \log k \rfloor +1$.
\label{cor:ht.comp}
\end{corollary}
\begin{proof}
To establish a competitive ratio of $\lceil \log_2 k \rceil + 1$ 
for {\sc R-ht}, we can directly apply the proof of Theorem~\ref{thm:huffman-approx}, along with the fact that the cost of the optimal HT for any scenario $s$ is at least 1. The lower bound follows by the same adversarial construction as in the proof of  Theorem~\ref{thm.ht.lb}, replacing the differences by the appropriate ratios at the conclusion of the proof.
\end{proof}
\section{Regret and Fairness in Binary Search Trees}
\label{sec:regret.bst}

In this section, we introduce the first study of fairness in BSTs, and demonstrate its connection to regret minimization in a {\em multiobjective} optimization setting. We begin with a motivating application. Consider a company that stores client information in a database structured as a BST. The clients are from either Spain or France, and the company would like to use a single database to store client data, instead of two, for simplicity of maintenance. Moreover, the company would like to treat customers for the two countries in a fair manner. Namely, the average cost for accessing clients from France should be comparable to that of accessing clients from Spain. 

We can formulate applications such as the above using a scenario-based regret-minimization framework over BSTs. Specifically, we can model the setting using two frequency vectors, one for each country. Each vector stores the probability of accessing a client, i.e., entry $i$ is the probability of accessing client $i$. Since we treat all clients of a country equally,  we are interested in tradeoffs between the average access costs of clients in the two countries, and can thus assume that the access probabilities are uniform. That is, if $f$ denotes the number of French clients, then each such client has access probability $1/f$ in the frequency vector. In the application we discuss, a client can belong to one of two countries only. This is not a required assumption, but one we can make without loss of generality. See Proposition~\ref{prop:wlog}.

\begin{proposition}
A worst case instance for the robust BST problem satisfies that every key $i$ is dedicated to a specific scenario, in the sense that it has a positive probability for at most one scenario.
\label{prop:wlog}
\end{proposition}
\begin{proof}
Let $I$ be an instance with 2 scenarios. Let $i$ be a key in $I$. 
Let $p^0, p^1$ be the distributions for the respective scenarios $\textbf 0$ and $\textbf 1$.

Let $I'$ be an instance obtained from $I$ by replacing key $i$ by new keys $i', i''$ and with the following probabilities
          
\[
\begin{array}{l|ll}
           & i'     &i''  \\ \hline
\textrm{scenario\:} \textbf{0} & p^0_i  &0 \\
\textrm{scenario\:} \textbf{1} & 0      &p^1_i
\end{array}
\]

Observe that in the optimal tree for scenario \textbf{0} in $I'$, the key $i''$ must be a leaf by optimality. Hence, it has the same cost as the optimal tree for scenario \textbf{0} in I. The same observation holds for scenario \textbf{1}.  

Let $T'$ be a BST which minimizes the competitive ratio for $I'$.
Suppose that the level of $i'$ is at most the level of $i''$ in $T'$. 
The proof for the other case is symmetric. 

Observe that $i''$ has no left descendant in $T'$. The reason is that this left descendant needs to be $i'$ or its ancestor by the key ordering. But $i'$ is at the same or lower level than $i''$.

Now transform $T'$ into a tree $T$ for instance $I$ as follows. Replace $i'$ by $i$. Remove $i''$ if it is a leaf, otherwise replace $i''$ by its unique descendant. The level of each key different from $i, i', i''$ is decreased by 1 for the descendants of $i''$ in $T'$ and preserved otherwise. In scenario \textbf{0} of instance $I'$, the keys $i'$ and $i''$ contribute with the value which is at most the value with which key $i$ contributes in scenario \textbf{0} of instance $I$. This is due by the level ordering assumption between $i'$ and $i''$. The same observation holds for scenario \textbf{1}. 

As a result in each scenario, the cost of $T'$ is at least the cost of $T$, which is at least the optimal cost. It follows that the competitive ratio for $I'$ is at least the competitive ratio for $I$. Hence the worst case instance for the competitive ratio satisfies without loss of generality that each key has a positive probability for exactly one scenario. The same observation holds for the cost minimization and the regret minimzation variants.
\end{proof}

There are two important observations to make. First, unlike the approaches of Sections~\ref{sec:bst} and~\ref{sec:huff}, the fairness setting is inherently multi-objective. For example, in the above application, we are interested in the tradeoff between the total access costs of clients in the two countries. We will thus rely on the well-known concept of {\em Pareto-optimality}~\cite{boyd2004convex} that allows us to quantify such tradeoffs, as we will discuss shortly. Second, we will use the case of two scenarios for simplicity, however we emphasize that the setting and our results generalize to multiple scenarios, as we discuss at the end of the section.

We formalize our setting as follows. We are given two scenarios {\bf 0} and {\bf 1} over $n$ keys $1, \ldots ,n$. We denote by $a$ and $b$ the number of keys in scenario $0$, and $1$, respectively.
We refer to keys of scenario 0 as the {\em $0$-keys}, and similarly for $1$-keys. We can describe compactly the two scenarios using a {\em binary string} $s\in\{0,1\}^n$, which specifies that key $i$ belongs to scenario $s_i \in\{0,1\}$. 
Consider a BST $T$ for this set of $n$ keys, then the cost of key $i\in[n]$ is the level of the node in $T$ that contains $i$.  The {\em $0$-cost} of $T$ is defined as the total cost of all $0$-keys in $T$, and the $1$-cost is defined similarly. 

To define the concept of regret, let 
$\textrm{OPT}(m)$ denote the optimal cost of a binary tree over $m$ keys, assuming a uniform key distribution. From~\cite[Sect 5.3.1, Eq. (3)]{knuth1997art}, 
\begin{equation}
\text{OPT}(m) = (m+1)\lceil\log_2(m+1)\rceil - 2^{\lceil\log_2(m+1)\rceil} + 1.
\label{eq:opt.bound}
\end{equation}

Clearly, in every BST $T$, the $0$-cost is at least $\textrm{OPT}(a)$. We refer to the difference between the $0$-cost of $T$ and the  quantity $\textrm{OPT}(a)$ as the {\em $0$-regret} of $T$. Thus, the $0$-regret captures the additional cost incurred for searching $0$-keys, due to the presence of $1$-keys in $T$ ($1$-regret is defined along the same lines). This notion allows us to establish formally the concept of fairness in a BST:

\begin{definition} \label{def:fairness}
A BST for a string $s$ is $(\alpha,\beta)$-fair if it has $0$-regret $\alpha$ and $1$-regret $\beta$. We call $(\alpha,\beta)$ the {\em regret point} of the BST. We denote by $f(s,\alpha)$
the function that determines the smallest $\beta$ such that there is a BST for $s$ which is $(\alpha,\beta)$-fair.
\end{definition}

\begin{figure}
    \centering
    \includegraphics[width=0.5\linewidth]{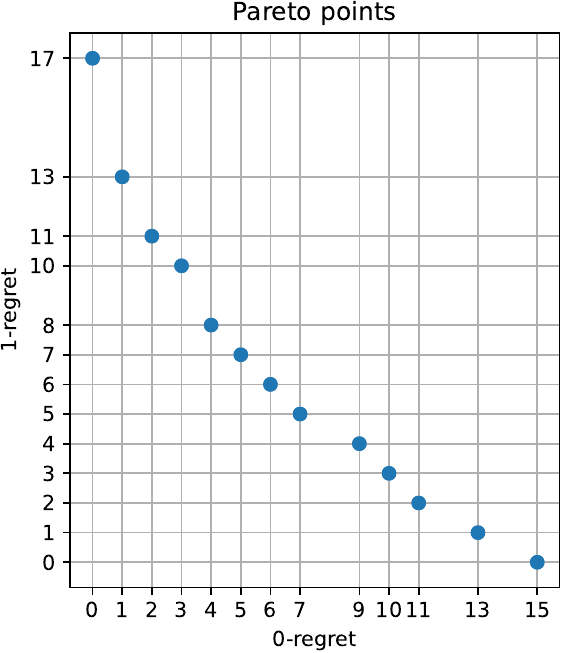}
    \caption{The Pareto optimal regret points for the string $1011011001111000$.}
    \label{fig:f}
\end{figure}

We say that a tree $T$ dominates a tree $T'$ if the regret point $(\alpha,\beta)$ of $T$ dominates the regret point $(\alpha', \beta')$ of $T'$ in the sense that $\alpha\leq \alpha'$ and $\beta\leq \beta'$ with one of the inequalities being strict.  The Pareto-front is comprised by the regret points of all undominated trees. This is illustrated in Figure~\ref{fig:f}.

\subsection{Computing the Pareto Frontier}
\label{subsec:fair.pareto}

We now give an algorithm for computing the Pareto frontier. More precisely, we describe how to compute the function $f(s,\alpha)$. Note that $f$ is non-increasing in $\alpha$. The Pareto front is obtained by calling $f$ for all $\alpha=0,1,\ldots$, until $f(s,\alpha)=0$. We first need to bound the range of $\alpha$.

\begin{lemma}
\label{lemma:smallest.alpha}
Let $\alpha^*$ denote the smallest integer such that  $f(s,\alpha^*)=0$. It holds 
that $\alpha^* \leq a\lfloor \log_2 (b+2)\rfloor$.
\end{lemma}
\begin{proof}
Any BST with zero $1$-regret consists of a complete tree for the scenario $1$ keys. Such a tree has depth $d=\lfloor \log_2 (b+1)\rfloor$, where $b$ is the number of 1s in $s$.
There are three types of nodes in this tree, depending on their out-degree: leaves, nodes with out-degree 1, and full nodes.
All nodes strictly below level $d-1$ are full. All nodes at level $d$ are leaves. At level $d-1$, there may be nodes with out-degree 1, moreover, there may be leaves at level $d-1$ if and only if $b\leq 2^d-3$.  Hence, a worst-case string $s$ is $10^a1^{b-1}$, for which the $0$-keys need to form a complete tree, and which is attached to the right of the left-most $1$-key. This $1$-key can be at level $d-1$ only if $b\leq 2^d-3$. As a result, $\alpha^*=a\lfloor \log_2 (b+2)\rfloor$.  Figure~\ref{fig:alpha-star} illustrates this situation.
\end{proof}

\begin{figure}[htb!]
    \centering
\begin{tikzpicture}[scale=1]
\foreach \rank/\x/\y in {
    1/12/0,
    2/8/1,
    3/16/1,
    4/0/2,
    5/10/2,
    6/14/2,
    7/18/2,
    9/4/3,
    10/9/3,
    11/11/3,
    12/13/3,
    13/15/3,
    14/17/3,
    15/19/3,
    18/2/4,
    19/6/4,
    36/1/5,
    37/3/5,
    38/5/5,
    39/7/5
    } 
    {
    \ifthenelse{\rank > 16 \OR \rank = 9} 
    {
        \draw (\x * 0.6, -\y) node (v\rank) {\textcolor{red}0};
    }
    {
        \draw (\x * 0.6, -\y) node (v\rank) {\textcolor{blue}1};
    };
    \ifthenelse{\rank > 1}
    {
        \pgfmathtruncatemacro{\ancestor}{\rank / 2};
        \draw (v\ancestor) -- (v\rank);
    }{};
    };
    \draw (-1, 0) node {level};
    \draw (-1, -2) node {$d-1$};
    \draw (-1, -3) node {$d$};
\end{tikzpicture}
    \caption{An illustration of the situation in the proof of Lemma~\ref{lemma:smallest.alpha}, for deriving $\alpha^*$.}
    \label{fig:alpha-star}
\end{figure}

Central to the computation of $f$ is the notion of \emph{loss}. For a given BST, we associate with each node, and for each scenario $c\in\{0,1\}$, a $c$-loss, such that the $c$-regret equals the total $c$-loss over all nodes. Informally, the $c$-loss at a node $r$ is the increase of the $c$-cost due to the choice of $r$ as the root of its subtree. Formally, we map $m_1,m_2\in\mathbb N$ and $m_0\in\{0,1\}$ to
\[
\textrm{loss}(m_1, m_0, m_2) = m_1+m_0+m_2 +\text{OPT}(m_1) + \text{OPT}(m_2) -\text{OPT}(m_1+m_0+m_2),
\]
and recall that $\text{OPT}$ is given by~\eqref{eq:opt.bound}.


The interpretation of this definition is the following. Consider a BST for a string $s\in\{0,1\}^n$ with node $r$ at its root. Let $m_1$ be the number of $c$ in the left sub-string $s[1:r-1]$, $m_2$ the number of $c$ in the right sub-string $[r+1:n]$, and $m_0$ the characteristic bit indicating whether  $s_r$ equals $c$. Then, if in both sub-trees scenario $c$ has zero regret, then its overall cost is exactly $\textrm{loss}(m_1,m_0,m_2)$.

We show how to compute the function $f$ by dynamic programming. While one could use the approach of~\cite{GIEGERICH2004215}, which computes the Pareto-optimal regret points for all BSTs for all sub-strings of $s$, we propose a somewhat different approach that has the same time complexity and is easier to implement.

The empty string $s=\varepsilon$ constitutes the base case for which we have $f(\varepsilon,\alpha) = 0$.
For $s \neq \varepsilon$ of length $n\geq 1$ we have
\[\begin{aligned}
    f(s,\alpha) = \min_{r} 
    \min_{\alpha_1, \alpha_2} 
    ( f(s[1,r-1],\alpha_1) +
    \textrm{loss}(b_1,b_0,b_2)+f(s[r+1,n],\alpha_2)
    ),
   \end{aligned}
\]
where the root $r$ in the outer minimization ranges in $[1,n]$ and separates the string to a left sub-string $s[1:r-1]$, a root $s_r$, and a right sub-string $s[r+1,n]$. For a fixed $r$, the value $a_1$ is the number of 0s in the left sub-string, whereas $a_2$ is the number of 0s in the right sub-string, $a_0$ is the indicator bit for $s_r=0$, and $b_1,b_0,b_2$ are similarly defined for scenario \textbf{1}. The inner minimization optimizes over all partitions $\alpha_1,\alpha_2$ of the allowed bound on the $0$-regret for the left and right sub-trees, such that $\alpha = \alpha_1 + \textrm{loss}(a_1,a_0,a_2) + \alpha_2$.

The correctness of the algorithm follows from the fact that any BST for $s$ with $0$-regret at most $\alpha$ is defined by a root and a partition of the remaining $0$-regret $\alpha-\textrm{loss}(a_1,a_0,a_2)$, and is composed recursively by a left and right sub-tree. 

For the running time of the algorithm, observe that there are $O(n^2)$ different sub-strings of $s$ and $O(n^2 \alpha^*)$  possible parameters for $f$, by Lemma~\ref{lemma:smallest.alpha}. Each is the minimization over $O(n)$ root candidates and $O(\alpha^*)$ regret bound partitions. The  overall time complexity is thus $O(n^3 (\alpha^*)^2)$, which simplifies to $O(n^5 \log^2 n)$ by Lemma~\ref{lemma:smallest.alpha}.

We emphasize that the dynamic programming approach can be generalized to $k\geq 2$ scenarios. In this general setting, a regret vector has dimension $k$, and in the function $f$ we fix $k-1$ dimensions and optimize the last one. Hence, we have $O(n^2 (n \log n)^{k-1}) = O(n^{1+k} \log^{k-1} n)$ variables, each being a minimization over $O(n (n\log n)^{k-1}) $ choices, which yields a time complexity of $O(n^{1+2k} \log^{2k-2} n)$.
\section{Computational Experiments} 
\label{sec:experiments}

\subsection{Robust BSTs and HTs}
\label{subsec:exp.first}

We report computational experiments on the robust versions of BSTs and HTs from Sections~\ref{sec:bst} and~\ref{sec:huff}. We used open data from~\cite{wiki:Letter_frequency}. 
Specifically, we chose ten European languages\footnote{
Danish, 
Dutch, 
English, 
Finnish,
French, 
German, 
Italian, 
Portuguese, 
Spanish and 
Swedish.} as corresponding to ten different scenarios, based on the frequency of each letter in the corresponding language. 
We restrict to the English alphabet of  26 letters, ignoring other letters or accents for simplification, but normalizing the frequencies to $1$. 
For example, the most frequent letter in English is \texttt{e} with 12.7\%, whereas in Portuguese the letter \texttt{a} is used more often with frequency 14.6\%. 

\begin{table}[ht!]
    \begin{center}
        \begin{tabular}{r|c|c|c|c|c|c|}
            \multicolumn{1}{l}{}& \multicolumn{2}{c}{Binary Search Tree} & \multicolumn{2}{c}{Huffman Tree} \\
            \multicolumn{1}{l}{} & \multicolumn{1}{c}{Optimal} & \multicolumn{1}{c}{{\sc R-bst}} & \multicolumn{1}{c}{Optimal} & \multicolumn{1}{c}{{\sc R-ht}}  \\\cline{2-5}
            cost   & 3.389 & 3.940 & 4.271 & 4.425  \\ \cline{2-5}
            competitive ratio & 1.047 & 1.215 & 1.038 & 1.091  \\ \cline{2-5}
            regret& 0.151 & 0.680 & 0.155 & 0.364  \\ \cline{2-5}
        \end{tabular}
    \end{center}
    \caption{Performance comparison of the various algorithms.}
    \label{tab:comparison}
\normalsize
\end{table}

To evaluate our algorithms, we provide a mixed integer linear program (MILP) formulation for the problems. This allows us to compute optimal trees with the help of commercial MILP solvers, such as \textsc{Gurobi}.
We give the MILP for cost-minimization in robust BSTs, but we note that minimizing the competitive ratio and the regret follow along the same lines, by only changing the objective function accordingly. The range of indices is $i,j,\ell\in[1,\ldots,n]$.
\begin{align}
    \textrm{minimize} \quad & C  \label{LP:obj}\\
    \textrm{subject to} \quad & \forall i: \: \sum_\ell x_{\ell,i} = 1 \label{LP:select} 
    \\
    & \forall i,j,\ell: \:  \sum_{r=i+1}^{j-1} \sum_{u=1}^{l-1} x_{u,r} \geqslant x_{\ell,i} + x_{\ell,j} - 1 \label{LP:BST}
     \\
    & \forall s: \:\sum_{\ell} \sum_i F^s_i\cdot \ell \cdot x_{\ell,i} \leqslant C \label{LP:ub}
     \\
  & \forall i,\ell:\: x_{\ell, i}\in \{0,1\}   
\end{align}

Here, $x_{\ell,i}$ is that $x_{\ell,i}=1$ if and only if key $i$ is assigned to level $\ell$. Constraint \eqref{LP:select} ensures that every key is assigned to exactly one level. Constraint \eqref{LP:BST} ensures that the resulting level vector satisfies Definition~\ref{level_definition}. Namely, we require that if keys $i$ and $j$ both have level $\ell$, then they are separated by a key at a lower level. Last, constraint \eqref{LP:ub} together with the objective \eqref{LP:obj} guarantee that $C$ is the maximum tree cost over all scenarios. 

The MILP for HTs can be obtaining using a similar process as for BSTs, with small differences. Namely, the levels $\ell$ start at $0$, and inequality \eqref{LP:BST} is replaced by a linear inequality ensuring that the resulting level vector is \emph{valid}, meaning that a prefix-free codeword can be associated with each key. For this purpose, we impose for all $\ell$,
\[
\sum_{b=0}^{\ell} 2^{\ell-b} \cdot |\{i: L_i = \ell\}| \leqslant 2^{\ell}.
\]
This can be rewritten as the following linear constraint, which replaces \eqref{LP:BST} in the MILP.
\[
\sum_{b=0}^{\ell} 2^{{\ell}-b} \sum_i x_{b,i} \leqslant 2^{\ell}.
\] 

We let Gurobi Optimizer version 9.1.2 solve our MILPs on a standard Windows 11 laptop with an Intel(R) Core(TM) i5-8365U CPU processor running at 1.60GHz with 16Gb of on-chip RAM. The observed running times, in seconds, are shown in the following Table~\ref{tab:gurobi.runtime}.

\begin{table}[ht!]
	\begin{center}
        \begin{tabular}{r|c|c|} \cline {2-3}
            & Binary Search Tree & Huffman Tree \\
            \cline {2-3}
            cost        &  8.23 & 0.07 \\ \cline {2-3}
            competitive ratio & 14.47 & 0.12  \\ \cline {2-3}
            regret      & 13.78 & 0.04\\ \cline {2-3}
        \end{tabular}
	\end{center}
\caption{\textsc{Gurobi} running times in seconds.}
\label{tab:gurobi.runtime}
\end{table}

In contrast, {\sc R-bst} and {\sc R-ht} run in less than $0.1$ seconds, using a Python implementation. Table~\ref{tab:comparison} summarizes the results of the experiments. As expected, the empirical performance is better than the worst-case guaranties of Theorems~\ref{th:alg.bst.ratio} and~\ref{thm:huffman-approx}. This is due to the fact that real data do not typically reflect adversarial scenarios (compare, e.g., to the adversarial constructions of Theorem~\ref{prop.bst.lb} and Theorem~\ref{thm.ht.lb}).

\subsection{Pareto-Optimality and Fairness}
\label{subsec:fair.experiments}

\begin{figure}[th!]
    \centering
    \includegraphics[width=0.8\textwidth]{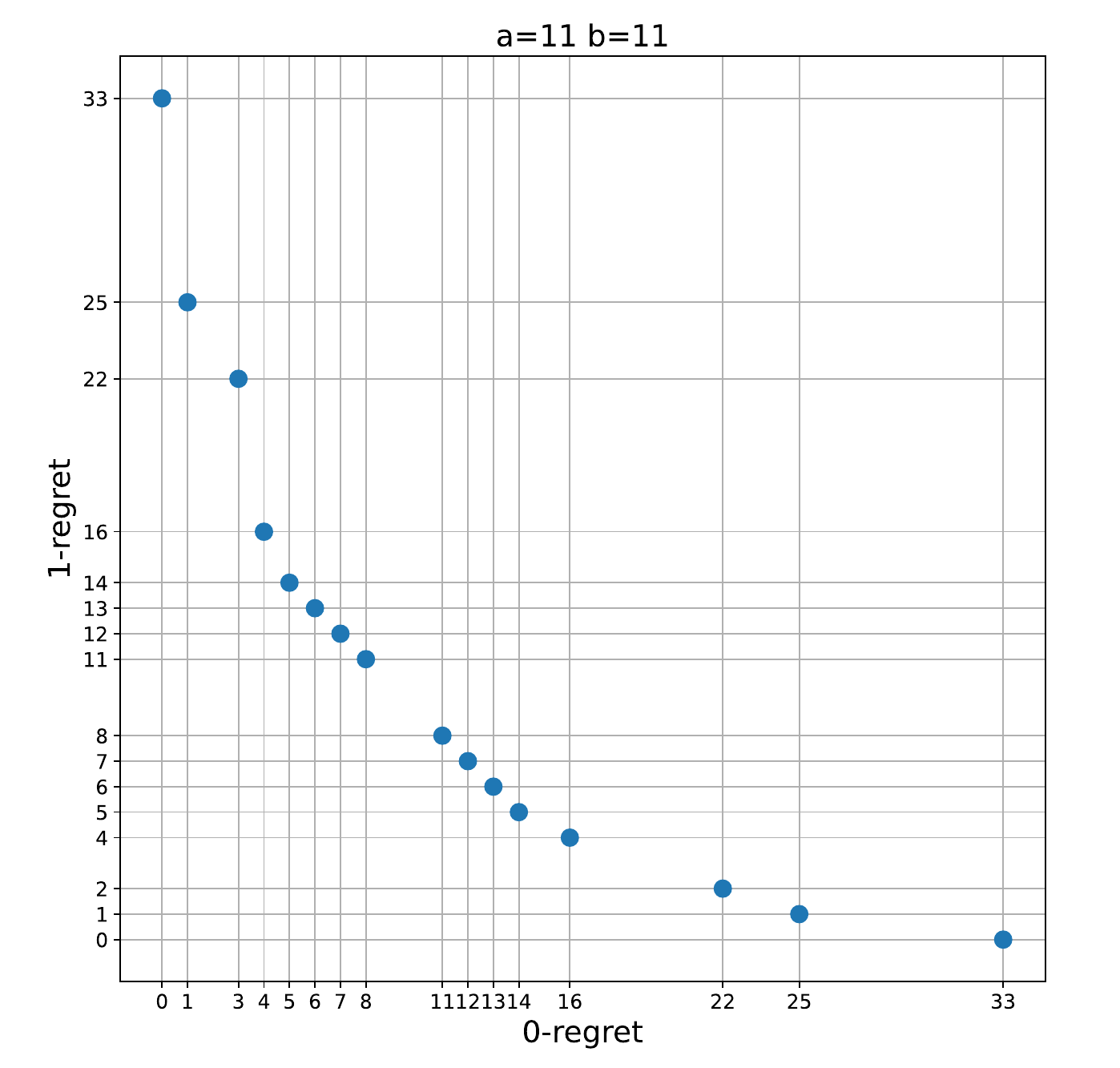}
    \caption{The Pareto front for strings with $a=11,b=11$.}
    \label{fig:g}
\end{figure}

We report experiments on our Pareto-optimal algorithm of Section~\ref{sec:regret.bst}, denoted by PO. The universe of all possible keys is represented by the names of cities from 10 different countries (the same counties as in Section~\ref{subsec:exp.first}). From this data, a string $s$ of length $2n$, for some chosen $n$, is generated by selecting two of the above countries (which we call country 0 and country 1) and the $n$ largest in population cities from each country. Then, $s$ is obtained by sorting lexicographically (i.e., alphabetically) the $2n$ cities and by setting $s_i \in \{0,1\}$, depending on whether the $i$-th city in $s$ belongs to country 0 or 1. 

We run algorithm PO and found the regret-based Pareto front for each string $s$ generated as above, choosing $n=30$. From the results, we observed that for every $s$, there exists a BST whose $0$-regret and $1$-regret are {\em both} bounded by $n$. Specifically, for every $s$, we were able to find a tree with a 0-regret of 27 and 1-regret of 28, as well as a tree with a 0-regret of 28 and a 1-regret of 27. These regret numbers were obtained for all $\binom{10}{2}$ country pairs. 

Furthermore, for all $a,b \in \{0,\ldots ,11\}$, we generated all strings $s$ of size $a+b$, i.e., $a$ keys from scenario \textbf{0} (country 0) and $b$ keys from scenario \textbf{1} (country 1). Figure~\ref{fig:g} illustrates our findings for the case $a=b=11$. Here, a point $(\alpha,\beta)$ signifies that we can find a BST of 0-regret $\alpha$ and 1-regret $\beta$ for {\em any} string  with 11 0s and 11 1s. This can be accomplished by running algorithm PO on all possible such strings. We obtained the same conclusion for all strings in which $a\leq 11$ and $b\leq 11$, as summarized  in the following observation.

\begin{remark}
The experimental results suggest that for a string $s$ consisting of $a$ 0s and $b$ 1s, there exists a BST of $0$-regret at most $a$ and $1$-regret at most $b$.
\label{rem:conjecture}
\end{remark}
\section{Conclusion}
\label{sec:conclusion}

We introduced the study of scenario-based robust optimization in data structures such as BSTs, and in data coding via Huffman trees. We gave hardness results, and theoretically optimal algorithms for a variety of measures such as competitive ratio, regret and Pareto-optimality. Our work also established connections between fairness and multi-objective regret minimization. Future work will address other data structures, such as B-trees and quad-trees, which play a predominent role  in data management. Our approaches address a fundamental issue: the tradeoff between cost and frequency of operations, which can be of use in many other practical domains, such as inventory management in a warehouse.

\printbibliography

\end{document}